\documentclass[aps,twocolumn,superscriptaddress,nofootinbib,pra,10pt]{revtex4-1}
  \usepackage{amsmath,amsfonts,amssymb,amsthm,graphicx,bbm,enumerate,times}
  \usepackage{mathtools}
  \usepackage[usenames,dvipsnames]{color}
  \usepackage[utf8]{inputenc}
  \usepackage{csquotes}

  \usepackage{hyperref}
  \mathtoolsset{showonlyrefs=true}

  \definecolor{jens}{rgb}{.2,0.7,.9}
  
  \definecolor{mathis}{rgb}{.9,.0,.9}
  
  \definecolor{albert}{rgb}{0.8,0.1,0.89}

  \newcommand{\ZZ}{\mathcal{Z}}
  \newcommand{\ZX}{\mathcal{Z}_X}
  \newcommand{\nn}{\mathbbm{N}}
  \newcommand{\rr}{\mathbbm{R}}
  \newtheorem{thm}{Theorem}
  
  \newtheorem{lm}{Lemma}
  
  \newtheorem{df}{Definition}

  \newtheorem*{lm*}{Lemma}
  
  \newtheorem*{prop*}{Proposition}
  \newtheorem*{res}{Result}
  \newtheorem{co}[thm]{Corollary}
  \newtheorem*{co*}{Corollary}
  \newcommand{\Id}{\mathbb{I}}
  \newcommand{\dmin}{d_{\mathrm{min}}}
  \newcommand{\dtrunc}{d_{\text{trunc}}}
  
  \newcommand{\dminw}{\widetilde{d}_{\mathrm{min}}}
  
  \newcommand{\ket}[1]{\left \vert #1 \right \rangle}
  \newcommand{\bra}[1]{\left \langle #1 \right \vert}
  \newcommand{\ketbra}[2]{\left \vert #1 \right \rangle \! \!\left \langle #2 \right \vert}
  \newcommand{\braket}[2]{\langle #1 \vert #2 \rangle}

  \DeclareMathOperator{\tr}{tr}
  
\begin{document}
\title{Local constants of motion imply information propagation}
  \author{M.\ Friesdorf}
  \affiliation{Dahlem Center for Complex Quantum Systems, Freie Universit{\"a}t Berlin, 14195 Berlin, Germany}
  \author{A.\ H.\ Werner}
  \affiliation{Dahlem Center for Complex Quantum Systems, Freie Universit{\"a}t Berlin, 14195 Berlin, Germany}
  \author{M.\ Goihl}
  \affiliation{Dahlem Center for Complex Quantum Systems, Freie Universit{\"a}t Berlin, 14195 Berlin, Germany}
  \author{J.\ Eisert}
  \affiliation{Dahlem Center for Complex Quantum Systems, Freie Universit{\"a}t Berlin, 14195 Berlin, Germany}
  \author{W.\ Brown}
  \affiliation{Computer Science Department, University College London, London WC1E 6BT}
  \date{\today}
\begin{abstract}
  Interacting quantum many-body systems are usually expected to thermalise, in the sense that
  the evolution of local expectation values approach a stationary value
  resembling a thermal ensemble.
  This intuition is notably contradicted in systems exhibiting many-body localisation,
  a phenomenon receiving significant recent attention.
  One of its most intriguing features is that, in stark contrast to the non-interacting case,
  entanglement of states grows without limit over time, albeit slowly.
  In this work,
  we establish a novel link between quantum information theory and notions of condensed matter,
  capturing the phenomenon in the Heisenberg picture.
  We show that the existence of local constants of motion, often taken as the defining
  property of many-body localisation, together with a generic spectrum,
  is sufficient to rigorously prove information propagation:
  These systems can be used to send a signal over arbitrary distances, in that
  the impact of a local perturbation can be detected arbitrarily far away.
  We perform a detailed perturbation analysis of quasi-local constants of motion
  and also show that they indeed can be used to construct efficient spectral tensor networks, as recently suggested.
  Our results provide a detailed and model-independent picture of information propagation in many-body localised systems.
\end{abstract}
\maketitle
  When driven out of equilibrium, interacting quantum many-body systems are usually expected to thermalise
  \cite{1408.5148,Polkovnikov_etal11,chrisRev},
  in the sense that local expectation values can be described by thermal ensembles.
  For this to be at all possible, local expectation values need to equilibrate to an apparent stationary state
  and energy has to be transported through the entire system.
  Such an expected generic behaviour is prominently violated by many-body localised systems \cite{basko2006metal}
  that show a strong suppression of transport
  \cite{Burrell,Pollmann_unbounded,1412.3073} and fail to serve as their own heat bath \cite{HuseReview,Integrable}.
  Thus, these systems do not thermalise and energy remains largely confined within certain regions.

  On the level of static properties of the Hamiltonian, equilibration in expectation
  is guaranteed by non-degenerate energy values and gaps
  \cite{Linden_etal09,PhysRevE.90.012121,reimann2012equilibration,chrisRev};
  a condition that is expected to hold with unit probability if small random interactions are added to the system.
  While these equilibration results are rather well understood,
  the question to what extend the local equilibrium values
  can be captured by thermal ensembles is still open to debate.
  A direct way to ensure thermal behaviour is given by the eigenstate thermalisation hypothesis
  \cite{Srednicki94,Deutsch91,Rigol_etal08}, one reading of which
  assumes that most individual eigenstates are already highly entangled and are locally indistinguishable from Gibbs states.

  For many-body localised systems, the static properties are markedly different. While the randomness typically occurring in
  these models will almost surely guarantee non-degenerate energy values and gaps, the individual eigenstates generically
  have low entanglement
  \cite{Integrable,Bauer,HuseReview} and are expected to be efficiently described in terms of tensor networks \cite{Bauer,1409.1252,1410.0687}.
  Moreover, one typically finds that the system has local constants of motion \cite{1407.8480,1404.5216,huse2014phenomenology} that are invariant in time.
  In fact, it has been shown that such local constants of motion can be used to infer the structure of the eigenstates and
  obtain an efficient tensor network description of the eigenprojectors \cite{1410.0687}.

  The investigation of information propagation in interacting many-body systems has a long tradition, with
  upper bounds, giving an effective speed of sound, being provided early on by Lieb and Robinson \cite{liebrobinson}.
  For localising systems, the non-interacting case notably leads to a
  full suppression of propagation, at least in the limit of infinite systems.
  It came as some surprise that this is no longer the case in the presence of interactions
  and that entanglement entropies very slowly grow without limit over time \cite{Pollmann_unbounded}.
  These numerical findings indicate that information is allowed to propagate in these models, at least
  for the infinite energy states usually considered, in a sense made more precise subsequently.

  In this work, we present a rigorous proof for information propagation in many-body localised systems,
  using remarkably few and innocent assumptions:
  only the existence of local constants of motions and a generic spectrum.
  Our approach is entirely model independent and assumes no specific structure of the Hamiltonian.
  This is achieved by basing the proof on a recently established link between spreading in the Heisenberg picture
  and equilibration behaviour \cite{gogolin,chrisRev}.
  Our results are a considerable step forward in the quest to prove that information is allowed to propagate in
  generic quantum many-body systems, which so far has only been achieved in highly specific systems.\\

\textit{Many-body localisation.}
  In this work, we will focus on systems exhibiting many-body localisation (MBL), which can be seen as a generalisation of
  Anderson localisation to interacting quantum many-body systems.
  While a comprehensive definition of this phenomenon is still lacking,
  it is generally expected that it is closely connected to the existence of approximately
  local constants of motion \cite{HuseReview,1410.0687}.
  These are operators that commute with the Hamiltonian
  \begin{align}
    [ H , \ZZ ] = 0,
  \end{align}
  but are nevertheless to some extent local.
  In order to access the locality of operators, we consider systems on a cubic lattice $\Lambda$
  of fixed dimension $\tilde{d}$, with a spin or fermionic degree of freedom at each site. Hence, the
  system's Hilbert space is given by $\otimes_{x\in\Lambda} \mathcal{H}_{loc}$, where $\mathcal{H}_{loc}$ is
  the Hilbert space of the local degree of freedom.
  We will denote the total number of sites by $L$ and local regions will be denoted by $S$ or $X$.
  The support of an operator is the region where it acts different from the identity.
  Of particular importance for our work are operators that are not strictly local, but only approximately local.
  To this end, it is convenient to introduce a map $\Gamma_S$, which restricts an operator $A$ to a region $S$
  \begin{align}
    \Gamma_{S} (A) := \Id_{S^c} \otimes d^{-|S^c|} \tr_{S^c} (A) \; ,
  \end{align}
  where $S^c = \Lambda \setminus S$ denotes the complement of $S$ and $\tr_{S^c}(A)$ denotes the partial trace of $A$ over $S^C$.
  The normalisation $d^{-|S^c|}$ is chosen such that the norm of operators that are already local on $S$ remains unchanged.
  In order to analyse the locality of an operator,
  we assume a central support region $X$, enlarged regions $X_l$
  and investigate how the approximation error scales with the size of the enlarged regions $X_l$.
  We choose the following description
  \begin{align}
    \|A - \Gamma_{X_l} (A) \| \leq
    \begin{cases}
      0 &\text{: \text{strictly local}},\\
      g(l) &\text{: \text{approximately local}},
    \end{cases}
  \end{align}
  where $g$ is some rapidly decaying function and $\| \cdot \|$ denotes the operator norm,
  amounting for Hermitian operators to the largest eigenvalue.
  Naturally, for strictly local observables, the initial region $X$ needs to be taken large enough
  to include the full support.

  In order to investigate the structure of the local constants of motion,
  we employ two simple models of many-body localisation.
  In one setting, it is assumed that the Hamiltonian is a sum of these commuting
  approximately local terms \cite{1412.3073}
  \begin{align}
    \label{H_kim}
    H = \sum_j \tilde{h}_j .
  \end{align}
  In the second setting, the Hamiltonian  is a higher order polynomial in terms
  of the approximately local constants of motion \cite{1404.5216,huse2014phenomenology}
  \begin{align}
    \label{H_huse}
    H = \sum_j h_i \ZZ_i + \sum_{i,j} J_{i,j} \ZZ_i \ZZ_j + \cdots,
  \end{align}
  where $J_{i,j}\in \rr$ for all $i,j$ and the interactions decay both with their order and with
  the distance of the involved spins.
  There is a very important difference between these two models.
  Many-body localisation is typically associated to a randomly chosen local potential as in the original
  work of Anderson \cite{Anderson}.
  Due to this disorder, it is strongly expected that the
  spectrum of the corresponding Hamiltonians is generic in the sense that it has non-degenerate energies and gaps.
  
  These assumptions give us already some information about the spectrum of the constants of motion.
  Since the Hamiltonian in Eq.\ \eqref{H_kim} is defined as a sum of commuting operators, it can only have a generic spectrum
  if each local constant of motion already has a generic spectrum.
  In contrast, the Hamiltonian in Eq.\ \eqref{H_huse} allows for higher order interactions in the constants of motion.
  Thus, each constant of motion might only have a small number of distinct eigenvalues similar to a simple local simple Pauli-Z-matrix,
  even though the total Hamiltonian has generic spectrum.

  This is the situation expected to occur in many-body localised systems and we will follow the intuition provided by the Hamiltonian with these higher order interaction between constants of motion,
  leading us to the following definition \cite{1412.3073}.
  \begin{df}[Local constant of motion]
    \label{df_IOM}
    Let $\ZZ$ be an operator that commutes with the Hamiltonian
    and has $M$ disjoint eigenvalues, all separated by a spectral gap lower bounded by $\gamma>0$,
    independent of the system size. $\ZZ$ is an exactly local constant of motion, iff it is strictly local and
    an approximately local constant of motion, iff it is approximately local.
  \end{df}
  The precise value of $M$ is not crucial for our purposes, since our results apply as long as the number $M$
  is independent of the system size.
  One direct consequence of the simple spectrum of the constants of motion is that the dimension of their eigenspaces has to grow exponentially in the system size.
  This can be seen from a perturbation theory point of view, where the exponentially small tails are not sufficient
  to create transitions between distinct eigenvalues. Thus the spectrum is approximately given by that of a
  strictly local operator, which has the feature of exponentially growing eigenspaces since it is extended by
  tensoring with identity (Sec.~\textbf{D} of SM \cite{supplement}).
  Moreover, using such constants of motion also allows to prove that
  the spectral tensor networks construction using exactly local constants of motion \cite{1410.0687}
  can still be carried out in the approximately local case (Sec.~\textbf{D} of SM \cite{supplement}).

  This leads to the interesting situation that for systems exhibiting MBL, eigenstates will typically
  efficiently be captured in terms of matrix-product states with low bond dimension.
  Whereas, product states will build up arbitrary large entanglement over time \cite{Pollmann_unbounded}.
  We now turn to our main result, namely a rigorous proof of information propagation in MBL systems.\\

\textit{Main result: Proof of information propagation.}
  In order to capture how information can be send through these models,
  we imagine that there are two parties, for brevity referred to as
  Alice and Bob, who have control over different parts of a spin system.
  For simplicity, let this model be 1D and let Alice control some part at the right end.
  We further assume a fixed separation between the parties and finally that Bob controls the rest of the chain.
  Alice encodes a message by either doing nothing or acting on his part of the spin chain with a local unitary $V$.
  At later time, Bob measures some local operator $A$. How well these two parties can communicate with such a protocol is
  captured by the difference in expectation value for Bob, depending on Alice's action,
  \begin{align}
    \bra{\psi} V A_t V^\dagger \ket{\psi} - \bra{\psi} A_t \ket{\psi}.
  \end{align}
  Such a procedure amounts to a positive channel capacity in the language of
  information theory.
  In less information theoretic terms: a local modification will necessarily
  lead to a measurable state modification far away in the chain for later times.
  At time zero, the support of $V$ and $A$ are spatially separated and the above quantity is zero.
  Over time, the support of $A_t$ might grow and thus eventually lead to a signal.
  Thus, whether the two parties can communicate crucially depends on the growth of an operator in the Heisenberg picture.
  In this way, the following quantity is a meaningful way to capture the capability of a Hamiltonian to propagate
  information.

  \begin{df}[Information propagation on average]
    \label{df_transport}
    A Hamiltonian allows for information propagation on average,
    if for any $\epsilon>0$, there exists a strictly local observable $A$ with unit operator norm,
    such that, for each finite region $S$, truncation of the Heisenberg evolution to that finite region necessarily
    leads to a fixed error $\varepsilon$, as long as the system size $L$ is large enough
    \begin{align}
      \forall \epsilon > 0 \; &\exists A \; \mathrm{ local } \; \forall S \; \exists L_0 \forall L > L_0: \\
                           &\overline{\| A_{t} - \Gamma_S (A_{t})\|} \geq 1 - \epsilon .
    \end{align}
    Here $\overline{A_t} = \lim_{T \rightarrow \infty} \frac{1}{T} \int_0^T \mathrm{d}t A_t$ denotes the infinite time average of $A$.
  \end{df}

  This definition is very restrictive, as it demands that the lower bound can be chosen arbitrarily close to $1$.
  On the other hand, it does not require any information on the corresponding time scale
  and allows for the support of the observable to take exponentially long to grow.
  From this definition, we can conclude
  \begin{align}
    |\bra{\psi} V A_t V^\dagger \ket{\psi} - \bra{\psi} A_t \ket{\psi}| > 1 - \epsilon \; ,
  \end{align}
  for some initial state vector $\ket{\psi}$ (Sec.~\textbf{C} of SM \cite{supplement}).
  Hence we can interpret $1-\epsilon$ as a signal strength.
  Thus, if a Hamiltonian allows for information propagation in the above sense, Alice and Bob
  can by the described protocol indeed communicate, no matter how large their separation,
  as long as Bob is allowed to measure on a large enough subsystem.
  Our main result states that this information propagation can be rigorously deduced
  from only the existence of a local constants of motion and a suitably non-degenerate spectrum of
  the Hamiltonian.

  \begin{thm}[Information propagation]
    Let $H$ be a Hamiltonian with non-degenerate energies and gaps
    and $\ZZ$ be a approximately local constant of motion
    with decay function $g$, localisation region $X$, spectral gap $\gamma>0$
    and eigenspaces with dimension larger than $\dminw$.
    Then $H$ necessarily has information propagation on average in the sense that
    there exists a local operator $A$ initially supported on $X_l\supset X$
    with $\|A\|=1$ such that $A_{t}$, on average, has support outside any finite region $S$,
    \begin{align}\label{eq:def_transp}
      \overline{\| A_{t} - \Gamma_{S} \left(A_{t} \right) \|}
      \geq 1 - 13 \frac{g(l)}{\gamma} - \frac{d_s}{2 {\dminw}^{1/2}}\; .
    \end{align}
  \end{thm}
  The first non-constant term in Eq.~\eqref{eq:def_transp} can be chosen arbitrarily small by picking the initial support $X_l$
  large enough and the second term decays exponentially with system size $L$, due to the
  growth of the degeneracy $\dminw$.
  This result shows that for many-body localising Hamiltonians, a zero velocity
  Lieb-Robinson bound does not occur and it is always possible to use the system to send information.
  Our results do not provide any statement on the influence of the involved energy scale and are perfectly
  compatible with the existence of a dynamical mobility edge, in the sense of a zero-velocity Lieb-Robinson
  bound for low energy states \cite{1409.1252,1411.0660}.
  Let it be noted that while MBL systems provably allow for information propagation, we expect that they
  do not exhibit particle or energy transport.

\textit{Proof idea: Equilibration implies information propagation.}
  Our results only rely on the existence of an approximately local constant of motion and assume no specific
  structure of the Hamiltonian.
  In order to first present the argument in its simplest form, however, we will use the following MBL model
  \begin{align}
    \label{H_huse}
    H = \sum_j h_j \sigma^z_j + \sum_{i,j} J_{i,j} \sigma^z_i \sigma^z_j,
  \end{align}
  where $\sigma^z_j$ are the local Pauli-Z-matrices and $J_{i,j} \in \rr$ decays exponentially with the distance between the spins.

  To carry out the proof in this simplified setting, we construct two objects.
  Firstly, a state that is the equal superposition of
  all eigenstates of the Hamiltonian and secondly, a local operator $A$ that initially has expectation value one with respect to this state,
  but at the same time has a zero diagonal in the energy eigenbasis and thus zero expectation value
  for the equilibrated infinite time average.
  Since equilibration guarantees that local expectation values are described by the infinite time average,
  this will allow us to conclude that the Heisenberg evolution of the operator $A$ has to be non-local.
  \begin{lm}[Diagonal Hamiltonians]
    \label{main_thm}
    Let $H$ be a diagonal Hamiltonian on a spin-1/2 lattice with non-degenerate eigenvalues and gaps.
    Let $A = \sigma^x_j$ be the Pauli-$X$-matrix supported on spin $j$.
    Then $H$ necessarily has information propagation on average in the sense
    that the operator $A_{t}$ has, on average, support outside any finite region $S$
    \begin{align}
      \overline{\| A_{t} - \Gamma_S \left(A_{t} \right) \|} \geq 1 - d^{|S| - N/2}.
    \end{align}
  \end{lm}
  \begin{proof}
  From the concrete form of the Hamiltonian in Eq. \eqref{H_huse}, we could
  calculate the time evolution of the $X$-operator and see that it acquires strings of $Z$-operators
  that sooner or later extend over the whole chain. In the following, we will show how this spreading
  behaviour can be derived when no specific Hamiltonian structure is used.
  For the argument, we use a state vector $\ket{\psi}$ that is initially a product with $\ket{+}$ on all sites,
  which also is the equal superposition of all eigenstates of the system.
  Since we assume that the Hamiltonian has non-degenerate energies, we know that
  the infinite time average of $\rho = \ketbra{\psi}{\psi}$ is diagonal, since all off-diagonal
  elements correspond to non-zero energy gaps and are dephased away.
  Moreover, as the diagonal is invariant under the time evolution, the time-averaged state $\omega$
  is the normalised identity matrix.
  Considering a subsystem $S$, we can use the non-degenerate energy gaps to employ
  known equilibration results \cite{Linden_etal09} for the expected deviation from the time average
  \begin{align}
    \label{eq_equil}
    \overline{\| \tr_{S^c} (\rho_{t} - d^{-N} \Id ) \|_1} \leq \frac{d_{\text{sys}}}{2 {d_{\text{eff}}}^{1/2}} \leq d^{|S| - N/2}.
  \end{align}
  Here $N$ is the total number of spins and the effective dimension counts how many eigenstates of the Hamiltonian
  are part of the state
  \begin{align}
    d_{\text{eff}} = \left(\sum_k |\braket{k}{\psi}|^4\right)^{-1}.
  \end{align}
  The above result states that, for most times, the reduced state of $\rho_t$ looks like the identity.
  Due to the way equilibration is proven, the results also applies to the inverse evolution $\rho_{-t}$
  \cite{Linden_etal09}.

  To investigate information propagates under the Hamiltonian, we look at the time evolution of an observable $A$
  consisting of a single Pauli-$X$-operator somewhere in the region $S$.
  The key trick is to use the initial expectation value and to insert time evolution operators
  \begin{align}\label{eq:tr_one}
    1 = \tr ( A_0 \rho_0 )= \tr (A_t \rho_{-t} ) .
  \end{align}
  Since we know that the equilibrated state is the normalised identity,
  the expectation value of any local traceless operator $B$ has to vanish on average
  \begin{align}
    \overline {\tr (B \rho_{-t})} = 0.
  \end{align}
  Since $A_0$ is traceless and the time-evolution leaves the trace invariant,
  the operator $A_{t}$ on average cannot be local anymore.
  More precisely, we have that
  \begin{align}
    &\overline{\|A_{t} - \Gamma_S(A_{t})\|}\\
    &= \overline{\| A_{t} - d^{|S|-N} \Id_{S^c} \otimes \tr_{S^c} (A_{t}) \|}\\
    &\geq \overline{|\tr (A_{t} \rho_{-t})  - d^{|S|-N} \tr_S \left(\rho^S_{-t} \tr_{S^c} (A_{t})\right)|},
  \end{align}
  where we have used that $|\tr{(A \rho)}|\leq \|A\| \| \rho\|_1$ and have defined
  $\rho^S = \tr_{S^c} (\rho)$. Next we use the inverse triangle inequality, Eq.~\eqref{eq:tr_one} and
  insert $0 = \tr (d^{-|S|} \Id_S \tr_{S^c}(A)) $ which is using the fact that
  the reduced observable has zero expectation value with the infinite time average
  \begin{align}
    &\overline{|\tr (A_{t} \rho_{-t})  - \tr_S \left(\rho^S_{-t} d^{|S|-N} \tr_{S^c} (A_{t})\right)|}\\
    &\geq 1 - \overline{|\tr_S \left(\rho^S_{-t} d^{|S|-N} \tr_{S^c} (A_{t})\right)|}\\
    &\geq 1 - \overline{|\tr_S \left( (\rho^S_{-t} - d^{-|S|} \Id_S ) d^{|S|-N} \tr_{S^c} (A_{t})\right)|} \; .
  \end{align}
  Another application of $|\tr{(A \rho)}|\leq \|A\| \| \rho\|_1$ allows
  to use the equilibration results discussed previously.
  Using $\|d^{|S|-N} \tr_{S^c} (A_{t})\| \leq 1$ concludes the estimate
  \begin{align}
    &\overline{\|A_{t} - \Gamma_S(A_{t})\|}\\
    &\geq 1  - \| \rho^S_{-t} - d^{-|S|} \Id_S \|_1 \|d^{|S|-N} \tr_{S^c} (A_{t})\| \\
    &\geq 1 - d^{|S|-N/2} \;.
  \end{align}
  \end{proof}

  The above proof explicitly establishes a recently proposed connection between propagation
  and equilibration \cite{gogolin,chrisRev}.
  Thus, the assumption of non-degenerate energy gaps is only needed to guarantee equilibration,
  which means that the condition can be substantially relaxed \cite{PhysRevE.90.012121}.

  The main idea for the proof still can be carried out in the setting where the Hamiltonian is no longer
  assumed to be diagonal, but where only the existence of an approximately local constant of motion is guaranteed.
  For this, it is first assumed that the constant of motion is exactly local. This implies that it
  is possible to distinguish different sets of eigenvectors locally and thus allows to
  construct local observables that have zero diagonal in the eigenbasis of the Hamiltonian.
  Moreover, a state with large expectation value with respect to this observable can be constructed,
  again allowing to use equilibration results, together with the off-diagonality
  of the observable in order to prove information propagation.
  Finally, a perturbation analysis extends the argument to approximately local constants of motion ec.~\textbf{C} of SM \cite{supplement}).

\textit{Discussion \& outlook.}
  In this work, we have shown that for systems with suitably non-degenerate spectrum a single
  approximately local constant of motion is sufficient to rigorously derive information propagation.
  We explicitly construct local excitation operators whose effect spreads over arbitrary distances,
  thus giving rise to a protocol for using MBL systems for signalling.
  This implies that the recently derived logarithmic light cone \cite{1412.3073}
  can never be tightened to a zero-velocity Lieb-Robinson bound, at least if one allows for infinite energy
  in the system. These results strengthen and are concomitant with the observation of a logarithmic
  growth of entanglement entropies in many-body localised models.

  As future work, it would be interesting to explore the speed of the information propagation,
  which naturally is linked to the open problem of deriving physically meaningful time scales of equilibration in local models.
  Another important question is how the information propagation derived above is linked to the available energy
  scale in the system. In particular, it would be interesting how our results relate to the
  possibility of having a mobility edge and how they are connected to the presumed suppression of
  energy and particle transport in MBL systems.

{\it Acknowledgements.}
  We thank Christian Gogolin for many interesting discussions on possible links between equilibration and transport and
  Henrik Wilming and Tobias Osborne for numerous discussions on many-body localisation.
  We are grateful for support by the EU (SIQS, RAQUEL, AQuS, COST), the ERC (TAQ),  the BMBF, and the Studienstiftung des
  Deutschen Volkes.
  WB is supported by the EPSRC.
\bibliographystyle{naturemag}

\section*{Supplemental Material}
\subsection{Local systems}
  In this appendix, we review some basic definitions for local quantum many-body systems.
  We  work with a finite lattice $\Lambda$ with $d$-dimensional spin systems attached to each vertex of the lattice.
  We look at local regions $X$ and denote by $|X|$ the number of spins contained in such a region.
  For any fixed set $X$, we will introduce enlarged sets $X_l$ that contain
  $X$ as well as all sites within distance $l$ of $X$.
  The distance measure will be the Manhatten metric, such that distances are always natural numbers.
  Local reductions of an observable will be performed by a map
  \begin{align}
    \Gamma_{S} (A) := d^{-|S^c|} \Id_{S^c} \otimes \tr_{S^c} A,
  \end{align}
  where $S^c$ denotes the complement of $S$.
  Due to the duality of operator and trace norm, we have
  \begin{align}
    \| \Gamma_S (A) \| \leq \|A\|
  \end{align}
  for any region $S$.
  In the following, we introduce local operators and unitaries.
  \begin{df}[Local observable]
    \label{def_local_observable}
    An operator $A$ will be called $(g,X)$-local, if there exists a finite localisation region $X$
    such that
    \begin{align}
      \|A - \Gamma_{X_l} (A) \| \leq \| A \| g(l)
    \end{align}
    for some function $g:\nn\rightarrow \rr$ with suitable decay in $l$.
  \end{df}
  \begin{df}[Local unitary]
    \label{def_local_unitary}
    A unitary operator $U$ will be called $f$-local, if the conjugation of a local observable $A$ with localization region $X$ remains
    local in the sense that
    \begin{align}
      \|U A U^\dagger - \Gamma_{X_l} (U A U^\dagger) \| \leq \| A \| f(l)
    \end{align}
    for some function $f:\nn\rightarrow \rr$ with suitable decay in $l$.
  \end{df}
  Correspondingly, we will say that a Hamiltonian has $f$-local eigenvectors, if the unitary diagonalising
  it is $f$-local.

  The main focus of this work is to investigate the propagation of local operators.
  We work with the following definition of information propagation that captures spreading of support in the Heisenberg picture.
  \setcounter{df}{1}
  \begin{df}[Information propagation on average]
    A Hamiltonian allows for information propagation on average,
    if for any $\epsilon>0$, there exists a strictly local observable $A$ with unit operator norm,
    such that, for each finite region $S$, truncation of the Heisenberg evolution to that finite region necessarily
    leads to a fixed error $\varepsilon$, as long as the system size $L$ is large enough
    \begin{align}
      \forall \epsilon > 0 \; &\exists A \; \mathrm{ local } \; \forall S \; \exists L_0 \forall L > L_0: \\
                           &\overline{\| A_{t} - \Gamma_S (A_{t})\|} \geq 1 - \epsilon .
    \end{align}
    Here $\overline{A_t} = \lim_{T \rightarrow \infty} \frac{1}{T} \int_0^T \mathrm{d}t A_t$ denotes the infinite time average of $A$.
  \end{df}
  \setcounter{df}{4}
  The form of the estimate is well known for Lieb-Robinson bounds that provide
  an upper bound for the above quantity depending on time
  \cite{liebrobinson,LRreviewchapter}. In this way, our definition can be seen as a
  \enquote{lower Lieb-Robinson bound}, albeit without any information on the corresponding time scale.
  Note that the above information propagation on average naturally also leads to information propagation
  for a large fraction of time maybe with respect to an adjusted error $\varepsilon$.

  As in the case of Lieb-Robinson bounds \cite{quant-ph/0603121}, this
  definition can be connected to information propagation in the lattice.
  In particular, we can rewrite the reduction map $\Gamma$ using an integration over unitaries drawn from
  the Haar measure
  \begin{align}
    \Gamma_S(A) = \int_{S^c} dU U A U^\dagger,
  \end{align}
  which gives
  \begin{align}
    \int_{S^c} dU \overline{\|[U, A] U^\dagger \|}
    &\geq
    \overline{\| \int_{S^c} dU [U, A] U^\dagger \|}\\
    &= \overline{\| A - \Gamma_S (A)\|} \geq 1 - \epsilon .
  \end{align}
  Thus for systems showing information propagation according to Def.~$2$ in the main text,
  we can find a state vector $\ket{\psi}$ and such that a local excitation created by some local unitary $V$ will be,
  on average, detectable at distances arbitrary far away
  \begin{align}
    \overline{\bra{\psi} V A_t V^\dagger \ket{\psi} - \bra{\psi} A_t \ket{\psi}}
    &= \overline{\bra{\psi} [V, A_t] V^\dagger \ket{\psi}} \\
    &> 1 - \epsilon.
  \end{align}
\subsection{Simple MBL model implies information propagation}
  \label{app_simple}
  In this appendix, we prove information propagation in two simple MBL models,
  one where the Hamiltonian is diagonal in the computational basis and one
  where the eigenstates are deformed by a f-local unitary (Definition \ref{def_local_unitary}).
  We start with the following simple model
  \begin{align}
    \label{H_huse}
    H = \sum_j h_i \sigma^z_i + \sum_{i,j} J_{i,j} \sigma^z_i \sigma^z_j,
  \end{align}
  which consists of interacting Pauli-Z-matrices with suitably random coefficients $J_{i,j}$
  that decay with the distance between sites $i$ and $j$.
  In this system, the eigenvectors are the computational basis and the energies and gaps will
  be non-degenerate due to the randomness in the model.
  For local dimension $d>2$, one can easily extend the model, by allowing for coupling with
  arbitrary local and diagonal matrices.
  In this case, we further require a generalised Pauli-X-matrix on site $j$ defined
  via matrix elements
  \begin{align}
    (\tilde{\sigma}^x_j)_{r,s} = \frac{1 - \delta_{r,s}}{d-1}.
  \end{align}
  As discussed in the main text, this model implies information propagation in the following way.
  \setcounter{lm}{0}
  \begin{lm}[Diagonal Hamiltonians]
    \label{main_thm}
    Let $H$ be a diagonal Hamiltonian with non-degenerate eigenvalues and gaps.
    Let $A = \tilde{\sigma}^x_j$ be the generalised Pauli-$X$-matrix supported on spin $j$.
    Then $H$ necessarily has information propagation on average in the sense
    that the operator $A_{t}$ has, on average, support outside any finite region $S$
    \begin{align}
      \overline{\| A_{t} - \Gamma_S \left(A_{t} \right) \|} \geq 1 - d^{|S| - N/2}.
    \end{align}
  \end{lm}

  The proof is contained in the main text and directly carries over to the case with local dimension $d>2$.
  Here, we will extent this result and show that
  the same construction can still be carried out in the case of approximately local eigenvectors,
  which are obtained from the computational basis by a joint $f$-local unitary
  (Def. \ref{def_local_unitary}).
  \setcounter{thm}{1}
  \begin{co}[$f$-local eigenvectors imply information propagation]
    \label{co_f}
    Let $H$ be a Hamiltonian with $f$-local eigenvectors and non-degenerate energies
    and gaps.
    Then $H$ necessarily has information propagation on average in the sense that for any fixed finite region
    $X_l$ of diameter $l$,
    there exists a local operator $A$ initially supported on $X_l$ with $\|A\|=1$
    such that $A_{t}$ has, on average, support outside any finite region $S$
    \begin{align}
      \overline{\| A_{t} - \Gamma_{S} \left(A_{t} \right) \|} \geq 1 - d^{|S| - N/2} - 2 f(l).
    \end{align}
  \end{co}
  \begin{proof}
  We will now use Lemma \ref{main_thm} in order to provide a proof for Corollary \ref{co_f} .
  For this, we use that the Hamiltonian can be diagonalised by a $f$-local unitary $V$
  and work with the observable
  \begin{align}
    A = V \tilde{\sigma}^x_j V^\dagger ,
  \end{align}
  where $\tilde{\sigma}^x_j$ is the generalised Pauli-$X$-matrix on some spin $j$ within the set $S$.
  This operator will no longer be strictly local, but due to the $f$-locality of the unitary $V$,
  the operator can be truncated
  \begin{align}
    \label{eq:f_trunc}
    \|A - \Gamma_{X_l} (A) \| \leq f(l) \;,
  \end{align}
  where $X_l$ denotes the set that contains the inital support, namely site $j$ and all $l$-nearest neighbours.
  Here we used that the operator norm of $A$ is one.
  From this, we can use the local reduction $A^l=\Gamma_{X_l}(A)$ as the local operator that will display information propagation.
  We will further need the time evolution of this truncated operator $A^l_{t} = e^{i t H} A^l e^{-i t H}$,
  where we first truncate and then evolve it in time.
  Naturally the unitary time evolution does not change the norm difference.
  The proof relies on a series of triangle inequalities. First we use that for any state
  \begin{align}
    \label{eq:f_twoterms}
    \|A^l_{t} - \Gamma_{S} (A^l_{t})\|
    \geq |\tr (A^l_{t} \rho_{-t})| - |\tr \left(\Gamma_S(A^l_{t}) \rho_{-t}\right)|\;.
  \end{align}
  Next we look at the two terms separately
  \begin{align}
    \label{eq:f_1term}
    |\tr (A^l_{t} \rho_{-t})|=
    &|\tr \left(e^{i t H} \Gamma_{X_l}(A) e^{-i t H} \rho_{-t}\right)|\\
    \geq& |\tr (A_{t} \rho_{-t})| - \|A - \Gamma_{X_l} (A) \|\\
    \geq& |\tr (A_0 \rho_0)| - f(l) = 1 - f(l),
  \end{align}
  where we have inserted a zero term $\pm \tr (A_{t} \rho_{-t})$ and have used
  the above truncation estimate in Eq. \eqref{eq:f_trunc} .
  The other term can be estimated as follows
  \begin{align}
    |\tr \left(\Gamma_S(A^l_{t}) \rho_{-t}\right)| =
    &|\tr \left(\Gamma_S (e^{i t H} \Gamma_{X}(A) e^{-i t H}) \rho_{-t}\right)|\\
    \leq& \|\Gamma_S (e^{i t H} (\Gamma_{X}(A) - A) e^{-i t H}) \|\\
        &+ |\tr_S \left ( \Gamma_S (A_{t}) \tr_{S^c} (\rho_{-t}) \right)| \; .
  \end{align}
  Here we again inserted a zero term
  \begin{equation}
    \pm \tr_S (\Gamma_S (A_{t}) ) \tr_{S^c} (\rho_{-t})
  \end{equation}
  and used a norm estimate.
  To proceed, we insert one more zero term
  $\pm \tr_S \left(\Gamma_S (A_{t}) d^{-|S|} \Id_S\right) $ and use the triangle inequality,
  \begin{align}
    &|\tr \Gamma_S(A^l_{t}) \rho_{-t}|\\
    &\leq \|\Gamma_{X}(A) - A\|
    + \|\tr_{S^c} (\rho_{-t}) - d^{-|S|} \Id_S \|_1\\
    &+ |\tr_S \left(\Gamma_S (A_{t}) d^{-|S|} \Id_S \right)| ,
  \end{align}
  and use that $\Gamma$ is a norm contractive map.
  These three terms can now easily be bounded. The first is small
  due to the $f$-locality of the unitary $V$ involved in constructing $A$, see Eq. \ref{eq:f_trunc} .
  The second term becomes small, once the time average is taken, which allows us to use
  known equilibration results \cite{Linden_etal09}
  \begin{align}
    \overline{ |\tr_S \left(\Gamma_S (A_{t}) d^{-|S|} \Id \right)|} \leq& d^{|S|-N/2} .
  \end{align}
  The third term vanishes completely, since the observable $A$ has zeros on its diagonal.
  This completes the estimate of the second term in Eq. \eqref{eq:f_twoterms}
  \begin{align}
    \label{eq:f_2term}
    \overline{|\tr \left ( \Gamma_S(A^l_{t}) \rho_{-t} \right)|} \leq f(l) + d^{|S|-N/2} \; .
  \end{align}
  Patching the estimates in Eqs. \eqref{eq:f_1term} and \eqref{eq:f_2term} together concludes the proof
  \begin{align}
    \|A^l_{t} - \Gamma_{S} (A^l_{t})\|
    \geq 1 - 2 f(l) - d^{|S|-N/2}\; .
  \end{align}

  \end{proof}
  The above Hamiltonians are special instances of systems having local constants of motion.
  In the diagonal case, the constants of motion are simply the local Pauli-Z-matrices.
  Once they are deformed by a quasi-local unitary, exact locality is lost, but one still
  obtains a full set of approximately local constants of motion.
\subsection{Constants of motion imply information propagation}
  \label{app_constant}

  In case the Hamiltonian has exactly local constants of motion \cite{1410.0687}
  $\ZX$ supported on some region $X$, they can also be employed
  to obtain the operators $A$ in the above construction.
  For this, let us assume that
  \begin{align}
    \ZX = \sum_{k=1}^M \lambda_k P_k
  \end{align}
  with exactly local projectors $P_k$ supported on $X$
  and $M$ distinct eigenvalues.
  The goal is then to construct an operator that is block-off-diagonal with respect to the projectors
  $P_k$. For this, let $\dmin$ be the smallest dimension of the eigenspaces of $\ZX$, when viewed as
  a local operator.
  For the construction, we fix two eigenspaces of $\ZX$.
  The larger of the two is then truncated down to the dimension $\dtrunc$ of the smaller one.
  Note that the resulting dimension of both spaces is lower bounded by $\dmin$.
  In these subspaces, we further fix some basis labelled by two indices $\ket{k,r}$
  where $k$ labels the eigenspaces of $\ZX$ and $r$ the basis vectors
  in each of these subspaces.
  We will denote the eigenspaces by $k=0$ and $k=1$.
  The operator $A$ is constructed to be supported on the small region $X$ and taken to be
  the flip operator between the subspaces
  \begin{align}
    A = \sum_{r}^{\dtrunc} \ketbra{k=0,r}{k=1,r} + \ketbra{k=1,r}{k=0,r}.
  \end{align}
  The operator norm of this observable is one and we will proceed by constructing an initial state that
  is an eigenstate of $A$ to eigenvalue 1, but still has large effective dimension.
  For this, we pick the subspace with smaller dimension and take the equal superposition,
  denoted by $\ket{v}$ of all eigenvectors in this subspace.
  For this, it is crucial to choose the subspace with smaller dimension, as the truncation
  in general, is not aligned with the eigenstates of the global Hamiltonian.
  The number of eigenvectors in the untruncated subspace will be lower bounded by
  $\dminw = \dmin d^{N-|X|}$,
  which is simply the smallest eigenspace dimension of $\ZX$ when viewed as an
  operator on the full lattice.
  The initial state vector is then taken to be
  \begin{align}
    \ket{\psi} = \frac{1}{\sqrt{2}} (\ket{v} + A \ket{v}).
  \end{align}
  It is straightforward to check that this is indeed an eigenstate of $A$, since $A^2 \ket{v} = \ket{v}$.
  What is more, the state vector $\ket{\psi}$ has an effective dimension lower bounded by $\dminw$.
  This gives the following result.

  \setcounter{thm}{2}
  \begin{co}[Information propagation: strictly local constants]
    Let $H$ be a Hamiltonian with non-degenerate energies and gaps
    and $\ZX$ be a strictly local constant of motion supported on $X$, with eigenspaces with dimension larger than $\dmin$.
    Then $H$ necessarily has information propagation in the sense that for any finite region $S$ containing $X$
    there exists a local operator $A$ initially supported on $X$ with $\|A\|=1$
    such that $A_{t}$, on average, has support outside $S$
    \begin{align}
      \overline{\| A_{t} - \Gamma_S \left(A_{t} \right) \|}
      \geq 1 - \frac{d^{|S| + |X|/2 }}{{\dmin}^{1/2}} d^{-N/2}.
    \end{align}
  \end{co}
  \begin{proof}
    To prove this statement, we can directly follow the proof of Lemma
    \ref{main_thm}, as presented in the main text.
    Using the construction of the initial state and the observable $A$ described above,
    we immediately obtain
    \begin{align}
      \overline{\| A_{t} - \Gamma_S \left(A_{t} \right) \|}
      \geq 1 - \frac{d_{\text{sys}}}{{d_{\text{eff}}}^{1/2}} ,
    \end{align}
    from the same equilibration results as in the main text.
    Inserting the  effective dimension described above $d_{\text{eff}} \geq \dminw = \dmin d^{N-|X|}$
    and $d_\text{sys}=d^{|S|}$ concludes the proof.
  \end{proof}

  In many localising systems, one does not expect the constants of motion to be strictly local, but only
  approximately local \cite{1410.0687}, meaning that $\ZZ$ is a $(g,X)$-local operator in the
  sense of Def.~$1$ in the main text.
  Using perturbation theory, it follows that this bound is sufficient to obtain local approximations
  for the eigenprojectors and makes it possible to once again construct an observable $A$
  that propagates through the system.

  \setcounter{thm}{0}
  \begin{thm}[Information propagation]\label{cot}
    Let $H$ be a Hamiltonian with non-degenerate energies and gaps
    and $\ZZ$ be a approximately local constant of motion with localisation region $X$,
    decay function $g$ with spectral gap $\gamma>0$
    and eigenspaces with dimension larger than $\dminw$.
    Then $H$ necessarily has information propagation on average in the sense that
    there exists a local operator $A$ initially supported on $X_l\supset X$
    with $\|A\|=1$ such that $A_{t}$, on average, has support outside any finite region $S$
    \begin{align}
      \overline{\| A_{t} - \Gamma_{S} \left(A_{t} \right) \|}
      \geq 1 - 13 \frac{g(l)}{\gamma} - \frac{d_s}{2 {\dminw}^{1/2}} \; .
    \end{align}
  \end{thm}
  Let us remark that the first term can be chosen arbitrarily small by picking the initial support $X_l$
  large enough and the second term decays exponentially with system size $L$, due to the
  growth of the degeneracy $\dminw$.

  Note that the role of the eigenspace dimension $\dminw$ is different to the previous corollary.
  Here, we look at the constant of motion as an operator on the full Hilbert space.
  We assume that the number of eigenvalues $M$ and the spectral gap are independent of the system size.
  From this, perturbation theory can be used to show that the eigenspace dimension $\dminw$
  has to grow exponentially in the system size.
  \begin{proof}
    The first step of the proof is to show that the approximate locality of the constant of motion
    also implies quasi-local eigenprojectors.
    Let
    \begin{align}
      \ZZ = \sum_{k=1}^M \lambda_k P_k
    \end{align}
    and let $\gamma$ denote the smallest spectral gap.
    Due to locality, we can express $\ZZ$ for each fixed $l$, as
    \begin{align}
      \ZZ &= \Gamma_{X_l}(\ZZ) + V_{l},
     \end{align}
    with a bounded perturbation $V_{l}$ satisfying
    \begin{align}
      V_{l} &= \Gamma_{X_l}(\ZZ) - \ZZ,\\
      \|V_{l}\| &< g(l) \;.
    \end{align}
    Let $P_k^{l}$ be the eigenprojectors for the truncated observable.
    Perturbation theory assures us that the perturbed eigenspaces stay approximately orthogonal
    (Theorem VII.3.1 in Ref. \cite{Bhatia-Matrix})
    \begin{align}
      \label{eq_proj}
      \| P_k (\Id - P_k^{l}) \| \leq \frac{\|V_{l}\|}{\gamma} = \frac{g(l)}{\gamma} \;,
    \end{align}
    which also implies
    \begin{align}
      \label{eq_proj2}
      \| P_k  - P_k^{l} \| \leq \frac{2 \|V_{l}\|}{\gamma} = \frac{2 g(l)}{\gamma} \;.
    \end{align}
    Choosing the distance $l$ large enough such that the function $g$ becomes smaller than ${\gamma}/{2}$,
    we know that the perturbed and unperturbed eigenspaces have the same dimension \cite{Bhatia-Matrix}.
    This local approximation of the eigenprojectors of the constant of motion will be the basis for the construction of
    the observable $A$ as well as the initial state $\rho$.

    In order to construct the observable, we will work with the truncated constant of motion $\Gamma_{X_l}(\ZZ)$,
    fix two subspaces and construct the same flip operator as in the case of exactly local constants of
    motion
    \begin{align}
      A = \sum_{r}^{\dtrunc} \ketbra{k=0,r}{k=1,r} + \ketbra{k=1,r}{k=0,r}.
    \end{align}
    Without loss of generality, let $k=0$ be the space with smaller dimension and $k=1$ the one
    truncated to $\dtrunc$.
    Let $P_1^l|_I$ be the projector on the truncated subspace of $P_1^l$ corresponding to the image of $A$.

    For the initial state, we will use the corresponding subspaces, again labelled by $k=0,1$ of the
    full constant of motion $\ZZ$. Again we pick the smaller of the two subspaces and define $\ket{v}$ to be the equal superposition
    of all eigenstates within this space.
    The initial state vector is then
    \begin{align}
      \ket{\psi} = \frac{1}{\sqrt{2}} (\ket{v} + A \ket{v}) \; .
    \end{align}
    By construction, the effective dimension and the equilibration results will be as in the case of a strictly
    local constant of motion.

    Crucial in the above construction is that we use the truncated constant of motion
    $\Gamma_{X_l}(\ZZ)$  for the observable $A$ in order to ensure locality, while we use the full object $\ZZ$ for the initial
    state in order to achieve a large effective dimension.
    What remains to be shown is that despite this locality difference in the construction,
    we still achieve a large expectation value of $A$ with $\ket{\psi}$, but an almost vanishing expectation
    value with the infinite time average.

    As a first step, we will show that $A$ is almost block-off-diagonal with respect to the eigenprojectors
    of the full constant of motion $\ZZ$.
    Introducing the identity $\Id = P_k^{l} + Q_k^{l}$, this takes the following form
    \begin{align}
      \|P_k A P_k\| \leq \| P_k P_k^{l} A Q_k^{l} P_k \| + \| P_k Q_k^{l} A P_k^{l} P_k \|
      \leq 2 \frac{g(l)}{\gamma}.
    \end{align}
    Here we used that $A$ is block-off-diagonal with respect to the truncated constant of motion $\Gamma_{X_l}(\ZZ)$.
    The same estimate holds for the projectors $Q_k$.

    Using this, bounding the expectation value with the infinite time average is straightforward
    \begin{align}
      \label{eq:qiom_w}
      \tr (A \omega) &= \tr (A P_0 \omega P_0) + \tr (A Q_0 \omega Q_0)\\
      &\leq \| P_0 A P_0 \| + \| Q_0 A Q_0 \| \leq 4 \frac{g(l)}{\gamma}.
    \end{align}

    We now have to show that the expectation value of $A$ with $\rho$ is large initially
    \begin{align}
      \bra{\psi} A \ket{\psi}
      &\geq \bra{v} A A \ket{v} - \frac{1}{2} |\bra{v} A \ket{v} + \bra{v} A A A\ket{v} | \; .
    \end{align}
    In the following, we will show that the first term is almost one, while
    the other two almost vanish due to the block-off-diagonality.
    For the first term, we will use that $A^2 = P_0^l + P_1^l|_I$,
    where $P_1^l|_I$ is the projector onto the image of $A$ in $P_1^l$.
    Using that $\bra{v} Q \ket{v}$ can only increase if we enlarge the subspace of the projector
    $Q$, we obtain
    \begin{align}
      \bra{v} A A \ket{v} &\geq \bra{v} P_0^l \ket{v} - |\bra{v} Q_0^l \ket{v}|\\
      &\geq \bra{v} P_0 \ket{v} - |\bra{v} Q_0 \ket{v}| - 4 \frac{g(l)}{\gamma} \\
      &= 1 - 4 \frac{g(l)}{\gamma},
    \end{align}
    where we used \eqref{eq_proj2}.
    The second term can be bounded directly using block-off-diagonality
    \begin{align}
      |\bra{v} A \ket{v}| \leq \|P_0 A P_0\| \leq 2 \frac{g(l)}{\gamma}.
    \end{align}
    The last term, finally can be bounded as follows.
    \begin{align}
      |\bra{v} A A A \ket{v}|
      &= |\bra{v} P_0 A (P_0^l + P_1^l|_I)  P_0 \ket{v}|\\
      &\leq \| P_0 A P_0^l\|  + \| P_1^l P_0\|\\
      &\leq \|P_0^l A P_0^l\| + 2 \frac{g(l)}{\gamma} + \frac{g(l)}{\gamma}\\
      &\leq 3 \frac{g(l)}{\gamma} .
    \end{align}
    To summarise, we have
    \begin{align}
      \label{eq:qiom_psi}
      \bra{\psi} A \ket{\psi} \geq 1 - 9 \frac{g(l)}{\gamma}.
    \end{align}
    Putting together the estimates for the expectation value of $A$ with the initial state, the equilibration result and the expectation
    value of $A$ with the infinite time average, we obtain the desired bound.
    More precisely, we choose $\rho = \ketbra{\psi}{\psi}$ and proceed as follows
    \begin{align}
      &\overline{\| A_{t} - \Gamma_{S} \left(A_{t} \right) \|}\\
      &\geq \tr (A \rho) - \overline{|\tr \left( \Gamma_{S} \left(A_{t_0} \right) \rho_{-t} \right) |}\\
      &\geq \tr (A \rho) - \overline{\|w - \rho_{-t}\|_1} - |\tr \left( \Gamma_{S} \left(A_{t_0} \right) w \right)|\\
      &\geq \tr (A \rho) - \frac{d_{\text{sys}}}{2{d_{\text{eff}}}^{1/2}} - |\tr \left( \Gamma_{S} \left(A_{t_0} \right) w \right)| \; .
    \end{align}
    Inserting the effective dimension $d_{\text{eff}}=\dminw$ and using
    Eqs. \eqref{eq:qiom_w} and \eqref{eq:qiom_psi}
    concludes the proof
    \begin{align}
      &\overline{\| A_{t} - \Gamma_{S} \left(A_{t} \right) \|}\\
      &\geq 1 - 9 \frac{g(l)}{\gamma} - \frac{d_s}{2 {\dminw}^{1/2}} - 4 \frac{g(l)}{\gamma}\\
      &\geq 1 - 13 \frac{g(l)}{\gamma} - \frac{d_s}{2 {\dminw}^{1/2}}.
    \end{align}
  \end{proof}
\subsection{Approximate spectral tensor networks}
  \label{app_vidal}
  In Ref.\ \cite{1410.0687}, it is shown that if a Hamiltonian has suitable local constants of motion,
  then each eigenprojector can be efficiently represented as a matrix product operator.
  Moreover, it is rigorously derived that there exists an efficient spectral tensor network for all eigenprojectors
  at the same time.
  Ref.\ \cite{1410.0687} then proceeds to sketch the case of approximately local constant of motion,
  for which similar conclusions are reached.
  In this appendix, we show that indeed, even for approximately local constants of motion with robust spectrum
  (Def.~$1$ in the main text), one can rigorously obtain a spectral tensor network.

  \begin{res}[Efficient spectral tensor networks]
    Let $H$ be a Hamiltonian with an extensive number of approximately local constants of motion (Def.~$1$ in the main text) with
  $|X|\leq L$ and $g(l)\leq c_1 \exp(-c_2 l)$, for suitable constants $c_1,c_2>0$.
  We assume that the approximately local constants of motion are algebraically independent,
  commute with each other and have suitable distributed support on the lattice.
  Then there exists an efficient spectral tensor network representation for all eigenprojectors of $H$.
  \end{res}

  The proof of this statement directly follows from Ref.\ \cite{1410.0687}, together with our Corollary \ref{cot}.
  We start from the observation that the approximate locality of the constant of motion also implies quasi-local eigenprojectors, in the sense that
  \begin{equation}
    \|P_k - P_k^l\|\leq \frac{2 g(l)}{\gamma},
  \end{equation}
  and that the perturbed and unperturbed eigenspaces have the same dimension.
  Using this approximation, one finds that projectors onto an  eigenspace of an approximately local constant of motion
  can be efficiently approximated by matrix-product operators.
  For a given site $j$, call $A$ the subset of sites for which the MPO approximations have a support that includes $j$.
  With the same argument as in Ref.\ \cite{1410.0687},
  choosing a path in the supports
  of the strictly local constants of motion and performing singular value decompositions, as outlined in Ref.\ \cite{1410.0687},
  one finds that the collection of all approximately local constants of motion in $A$ can again be written as a matrix-product operator.
  The stability Lemma \ref{lm_stab} below concludes the proof.

  \begin{lm}[Stability]
    \label{lm_stab}
    Let $\{\mathcal{Z}_j\}$ be a set of $N$ approximately local constants of motion with a lower uniform bound $\gamma>0$ on their minimal spectral gaps and  uniform upper bounds $L$ and $g(l)$ on size and decay of their localisation regions $X_j$, such that
    \begin{align}
      \|\mathcal{Z}_j - \Gamma_{X_j^l \supset X_j}(\mathcal{Z}_j)\|\leq g(l)
    \end{align}
    for any $X_j^l$ containing $X_j$ together with a buffer region of size $l$. Then if $P_{j,m}$ denotes the eigenprojector of $\mathcal{Z}_j$ for eigenvalue $m$ we have
    \begin{align}
      \|P_{j_1,m_1}\cdots P_{j_N,m_n} - P^l_{j_1,m_1} \cdots  P^l_{j_N,N_1}\| \leq 2 N \frac{g(l)}{\gamma}
    \end{align}
    with $P^l_{j_i,m_i} = \Gamma_{X_j^l}(P_{j_i,m_i})$ being strictly local.
  \end{lm}
  \begin{proof}
    The proof utilises perturbation theory on the level of the single eigenprojectors $P^l_{j_i,m_i}$
    similar to the proof of Corollary \ref{cot}.
    Using the triangle inequality we can upper bound the norm difference as
    \begin{align}
      &&\|P_{j_1,m_1}\cdots P_{j_N,m_n} - P^l_{j_1,m_1} \cdots  P^l_{j_N,N_1}\|\nonumber\\
   &&   \leq \sum_{k=1}^N \|P_{j_k,m_k}-P^l_{j_k,m_k}\| \; .
    \end{align}
    The result now follows from Eq.~\eqref{eq_proj2} and our uniformity assumptions.
  \end{proof}

\end{document}